\documentclass[table]{article} 
\usepackage{nips13submit_e,times}
\usepackage{hyperref}
\usepackage{amsfonts}
\usepackage{amsmath}
\usepackage{amssymb}
\usepackage{amstext}
\usepackage{graphicx}
\usepackage{picins}
\usepackage{paralist}
\usepackage{url}
\usepackage{amsthm}
\usepackage{wrapfig}
\usepackage{subfig,epsfig}
\usepackage{caption}
\usepackage{multirow}
\usepackage{xcolor}

\newcommand{\R}{\ensuremath{\mathbb R}}
\newcommand{\PB}{\ensuremath{\mathbf P}}

\makeatletter
\newcommand{\Rmnum}[1]{\expandafter\@slowromancap\romannumeral #1@}

\def\iid{i.i.d.}

\newtheorem{theorem}{Theorem}[section]

\theoremstyle{definition}

\theoremstyle{remark}

\title{Speeding up Permutation Testing in Neuroimaging \thanks{Hinrichs and Ithapu are joint first authors and contributed equally to this work.}}
\author{ Chris Hinrichs$^\dagger$ \hspace{1mm} Vamsi K. Ithapu$^\dagger$ \hspace{1mm} Qinyuan Sun$^\dagger$ \hspace{1mm} Sterling C. Johnson$^{\mathsection,\dagger}$\hspace{1mm} Vikas Singh$^\dagger$\\ 
  $^\mathsection$William S. Middleton Memorial VA Hospital \hspace{2mm} $^\dagger$University of Wisconsin--Madison\\
  {\footnotesize \texttt{\{hinrichs,vamsi\}@cs.wisc.edu} \hspace{2mm} \texttt{\{qsun28\}@wisc.edu}} \\ 
     {\footnotesize \texttt{\{scj\}@medicine.wisc.edu} \hspace{2mm}  \texttt{\{vsingh\}@biostat.wisc.edu}} \\
  {\footnotesize \url{http://pages.cs.wisc.edu/~vamsi/pt_fast}}}

\nipsfinalcopy 

\begin{document}

\maketitle

\begin{abstract}
Multiple hypothesis testing is a significant problem in nearly all neuroimaging
studies. In order to correct for this phenomena, we require
a reliable estimate of the Family-Wise Error Rate (FWER).
The well known Bonferroni correction method, while simple to implement, is quite conservative, and 
can substantially under-power a study because it ignores dependencies between test statistics.
Permutation testing, on the other hand, is an exact, non-parametric method of
estimating the FWER for a given $\alpha$-threshold, but for acceptably low
thresholds the computational burden can be prohibitive. In this paper, we show that permutation testing in fact amounts to populating the columns of a very
large matrix $\PB$.
By analyzing the spectrum of this matrix, under certain conditions, we
see that $\PB$ has a low-rank plus a low-variance residual decomposition which makes
it suitable for highly sub--sampled --- on the order of 0.5\% --- matrix completion
methods. 
Based on this observation, we propose a novel permutation testing methodology which offers
a large speedup, without sacrificing the fidelity of the estimated
FWER. Our evaluations on four different neuroimaging datasets show that a
computational speedup factor of roughly $50\times$ can be achieved while
recovering the FWER distribution up to very high accuracy.
Further, we show that the estimated $\alpha$-threshold is also recovered
faithfully, and is stable.
\end{abstract}

\section{Introduction}

Suppose we have completed a placebo-controlled clinical trial of a
promising new drug for a neurodegenerative disorder such as Alzheimer's disease
(AD) on a small sized cohort. The study is
designed such that in addition to assessing improvements in standard cognitive
outcomes (e.g., MMSE), the purported treatment effects will also be assessed
using Neuroimaging data. The rationale here is that, even if the drug does induce variations in cognitive symptoms, 
the brain changes are observable {\em much earlier} in the imaging data. 
On the imaging front, this analysis checks for statistically significant differences between brain images of subjects assigned to the two
trial arms: treatment and placebo.
Alternatively, consider a second scenario where we have completed a neuroimaging
research study of a particular controlled factor, such as genotype, and the interest is to evaluate {\em group-wise} differences in the
brain images: to identify which regions are affected as a function of class
membership. 
In either cases, the standard image processing workflow yields for each subject 
a 3-D image (or voxel-wise ``map''). 
Depending on the image modality acquired, these maps are of cerebral gray matter density,
longitudinal deformation (local growth or contraction) or metabolism.
It is assumed that these maps have been `co-registered' across different subjects 
so that each voxel corresponds to approximately the same anatomical location. \cite{ashburner2000,ashburner2001}.
%

In order to {\em localize} the effect under investigation (i.e., treatment or genotype), we then have to calculate a very large number (say,
$v$) of univariate voxel-wise statistics 
-- typically up to several million voxels. For example, consider group-contrast
$t$-statistics (here we will mainly consider $t$-statistics, however other
test statistics are also applicable, such as the $F$ statistic used in ANOVA testing, Pearson's
correlation as used in functional imaging studies, or the $\chi^2$
test of dependence between variates, so long as 
certain conditions described in Section \ref{sect-main-ideas} are satisfied).
In some voxels, it may turn out that a group-level effect has been indicated,
but it is not clear right away what its true significance level should be, if
any. As one might expect, given the number of hypotheses tests $v$, multiple testing
issues in this setting are quite severe, making it difficult to assess the true
Family-Wise Type I Error Rate (FWER) \cite{westfall1993}. If we were to address
this issue via Bonferroni correction \cite{bland1995}, the enormous number of
separate tests implies that certain weaker signals will almost certainly never
be detected, even if they are real. This directly affects studies of
neurodegenerative disorders in which atrophy proceeds at a very slow rate and
the therapeutic effects of a drug is likely to be mild to moderate anyway. 
This is a critical bottleneck which makes localizing real, albeit slight,
short-term treatment effects problematic. Already, this restriction will prevent
us from using a smaller sized study (fewer subjects), increasing the cost of
pharmaceutical research. In the worst case, an otherwise real treatment effect
of a drug may not survive correction, and the trial may be deemed a failure.

{\bf Bonferroni versus true FWER threshold.} Observe that theoretically, there
{\em is} a case in which the Bonferroni corrected threshold is close to
the true FWER threshold: when point-wise statistics are \iid ~If so, then 
the extremely low Bonferroni corrected $\alpha$-threshold crossings effectively become
mutually exclusive, which makes the Union Bound (on which Bonferroni correction
is based) nearly tight. However, when variables are highly {\em dependent} -- and indeed 
even without smoothing there are many sources of strong non-Gaussian dependencies between voxels, the true FWER threshold can be much
more relaxed, and it is precisely this phenomenon which drives the search for
alternatives to Bonferroni correction. Thus, many methods have been developed to
more accurately and efficiently estimate or approximate the FWER \cite{li2005,
storey, finner2009, leek2008}, which is a subject of much interest in statistics
\cite{clarke2009}, machine learning \cite{garcia2010}, bioinformatics
\cite{ge2003resampling}, and neuroimaging \cite{nichols2003}.

{\bf Permutation testing.} A commonly used method of directly and non-parametrically estimating the
FWER is Permutation testing \cite{nichols2003,singh2003}, which is a method of
sampling from the Global (i.e., Family-Wise) Null distribution.
Permutation testing ensures that any relevant dependencies present in the data carry
through to the test statistics, giving an unbiased estimator of the FWER.
If we want to choose a threshold sufficient to exclude {\em
all} spurious results with probability $1 - \alpha$, we can construct a
histogram of sample maxima taken from permutation samples, and choose a
threshold giving the $1 - \alpha/2$ quantile. Unfortunately, reliable FWER estimates
derived via permutation testing come at excessive (and often infeasible) computational cost --
often tens of thousands or even millions of permutation samples are required,
each of which requires a complete pass over the entire data set. This step alone
can run from a few days up to many weeks and even longer \cite{pantazis, gaonkar2013analytic}.

Observe that the very same dependencies between voxels, that forced the usage of permutation testing, indicate that 
the overwhelming majority of work in computing so many highly correlated Null statistics is redundant.
Note that regardless of their description, strong dependencies of almost any kind 
will tend to concentrate most of their co-variation into a low-rank subspace,
leaving a high-rank, low-variance residual \cite{li2005}. In fact, for Genome
wide Association studies (GWAS), many strategies calculate the `effective
number' ($M_{\rm eff}$) of independent tests corresponding to the rank of this
subspace \cite{cheverud2001,li2005}.
This paper is based on the observation that such a low-rank structure must also
appear in permutation test samples. 
Using ideas from online low-rank matrix completion \cite{he2012} we can
sample a few of the Null statistics and reconstruct the remainder as long as we
properly account for the residual. This allows us to sub-sample at {\em
extremely low rates}, generally $<1\%$. The {\bf contribution} of our work is to
significantly speed up permutation testing in neuroimaging, delivering running
time improvements of up to $50 \times$. 
In other words, our algorithm does the same job as
permutation testing, but takes anywhere from a few minutes up to a few
hours, rather than days or weeks. Further, based on recent work in random
matrix theory, we provide an analysis which sheds additional light on the use of
matrix completion methods in this context. 
To ensure that our conclusions are not an artifact of a specific dataset, we
present strong empirical evidence via evaluations on four separate neuroimaging
datasets of Alzheimer's disease (AD) and Mild Cognitive Impairment (MCI)
patients as well as cognitively healthy age-matched controls (CN),  showing that
the proposed method can recover highly faithful Global Null distributions, while
offering substantial speedups.

\section{The Proposed Algorithm}
\label{sect-background}
We first cover some basic concepts underlying permutation testing and low rank matrix completion in more detail, before presenting 
our algorithm and the associated analysis.  

\subsection{Permutation testing}

Randomly sampled permutation testing \cite{dwass1957} is a methodology for
drawing samples under the Global (Family-Wise) Null hypothesis. Recall that although
point-wise test statistics have well characterized univariate Null
distributions, the sample maximum usually has no analytic form due to the strong correlations across voxels.
Permutation is particularly desirable in this setting because it is free of any distribution assumption
whatsoever \cite{nichols2003}.
The basic idea of permutation testing is very simple, yet extremely powerful. Suppose
we have a set of labeled high dimensional data points, and a univariate test
statistic which measures some interaction between labeled groups for every dimension (or feature).
If we randomly permute the labels and recalculate each test statistic,
then by construction we get a sample from the Global Null distribution. 
The maximum over all of these statistics for every permutation sample is then used to 
construct a histogram, which therefore is a non-parametric estimate of the distribution of the sample maximum of
Null statistics. For a test statistic derived from
the real labels, the FWER corrected $p$-value is then equal
to the fraction of permutation samples which were {\em more extreme}. Note that all of
the permutation samples can be assembled into a matrix ${\bf P} \in \R^{v\times T}$
where $v$ is the number of comparisons (voxels for images), and $T$ is the number of permutation samples.

There is a drawback to this approach, however. Observe that it is in the nature of random sampling
methods that we get many samples from near the mode(s) of the
distribution of interest, but fewer from the tails.
Hence, to characterize the threshold for a small portion of the tail of this
distribution, we must draw a very large number of samples just so that the estimate converges. 
Thus, if we want an $\alpha = 0.01$ threshold from the Null sample maximum distribution, we
require many thousands of permutation samples --- each requires
randomizing the labels and recalculating all test statistics, a very computationally expensive procedure when $v$ is large. 
To be certain, we would like to ensure an especially low FWER by first setting $\alpha$
very low, {\em and then} getting a very precise estimate of the corresponding
threshold. The smallest possible $p$-value we can derive this way
is $1/T$, so for very low $p$-values, $T$ must be very large.

\subsection{Low-rank Matrix completion}

Low-rank matrix completion \cite{candes2010}
seeks to reconstruct missing entries from a matrix, given only a small
fraction of its entries. The problem is ill-posed unless we assume this matrix
has a low-rank column space. If so, then a much smaller
number of observations, on the order of $r\log(v)$, where $r$ is the column
space's rank, and $v$ is its ambient dimension \cite{candes2010} is
sufficient to recover both an orthogonal basis for the row space as well as 
the expansion coefficients for each column, giving the recovery. By placing an
$\ell_1$-norm penalty on the eigenvalues of the recovered matrix via the nuclear
norm \cite{fazel2004,recht2007} we can ensure that the solution is as low rank as
possible. Alternatively, we can specify a rank $r$ ahead of time, and estimate
an orthogonal basis of that rank by following a gradient along the Grassmannian
manifold \cite{balzano2010,he2012}. Denoting the set of randomly subsampled
entries as $\Omega$, the matrix completion problem is given as, 

{\small \begin{equation}
    \label{eq:matrix-completion}
    \min_{\tilde{\bf P}} \|{\bf P}_\Omega - \tilde{\bf P}_\Omega \|_F^2 \qquad
        \mbox{s.t.} \; \tilde{\bf P} = {\bf UW}; \; {\bf U} ~ \mbox{is orthogonal}
\end{equation} }
\noindent where ${\bf U} \in \R^{v\times r}$ is the low-rank basis of ${\bf P}$, $\Omega$ gives the measured entries, and
${\bf W}$ is the set of expansion coefficients which reconstructs $\tilde{\bf P}$ in
${\bf U}$. Two recent methods operate in an online setting, i.e., where 
rows of ${\bf P}$ arrive one at a time, and both ${\bf U}$ and ${\bf W}$ are updated
accordingly \cite{balzano2010,he2012}.

\subsection{Low rank plus a long tail}
\label{sect-main-ideas}
Real-world data often have a dominant low-rank component. 
While the data may not be {\em exactly} characterized by a low-rank basis,
the residual will not significantly alter the eigen-spectrum of the sample
covariance in such cases. Having
strong correlations is nearly synonymous with having a skewed eigen-spectrum,
because the flatter the eigen-spectrum becomes, the sparser the resulting
covariance matrix tends to be (the ``uncertainty principle'' between low-rank and
sparse matrices \cite{chandrasekaran2011}).
This low-rank structure carries through for purely linear statistics (such as sample means).
However, non-linearities in the test statistic calculation, e.g.,
normalizing by pooled variances, will contribute a long tail of eigenvalues,
and so we require that this long tail
will either decay rapidly, or that it does not overlap with the dominant
eigenvalues. For $t$-statistics, the pooled variances are unlikely to
change very much from one permutation sample to another (barring outliers) --- 
hence we expect that the spectrum of $\PB$ will resemble that of the
data covariance, with the addition of a long, exponentially decaying tail.
More generally, if the non-linearity does not de-correlate the test statistics
too much, it will preserve the low-rank structure.

If this long tail is indeed dominated by the low-rank structure, then its
contribution to $\PB$ can be modeled as a low variance Gaussian \iid ~residual. A
Central Limit argument 
appeals to the number of independent eigenfunctions that contribute to this
residual, and, the orthogonality of eigenfunctions implies that as more of
them meaningfully contribute to each entry in the residual, the more independent
those entries become. In other words, if this long tail begins at a low
magnitude and decays slowly, then we can treat it as a Gaussian \iid ~residual;
and if it decays rapidly, then the residual will perhaps be less Gaussian, but
also more negligible. Thus, our development in the next section
makes no direct assumption about these eigenvalues themselves, but rather that
the residual corresponds to a low-variance \iid~Gaussian random matrix --- 
its contribution to the covariance of test statistics will be Wishart
distributed, and from that we can characterize its eigenvalues.

\subsection{Our Method}
It still remains to model the residual numerically. By 
sub-sampling we can reconstruct the low-rank portion of ${\bf P}$ via matrix
completion, but in order to obtain the desired sample maximum distribution we
must also recover the residual. Exact recovery of the residual is essentially
impossible; fortunately, for our purposes we need only need its effect on the
distribution of the {\em maximum per permutation test}. So, we estimate its
variance, (its mean is zero by assumption,) and then randomly sample
from that distribution to recover the unobserved remainder of the matrix.

A large component in the running time of online subspace
tracking algorithms is spent in updating the basis set ${\bf U}$; yet, once a good
estimate for ${\bf U}$ has been found this becomes superfluous. We therefore divide
the entire process into two steps: training, and recovery. During the training
phase we conduct a small number of fully sampled permutation tests ($100$
permutations in our experiments). From these permutation tests, we estimate ${\bf U}$
using sub-sampled matrix completion methods \cite{balzano2010,he2012}, making
multiple passes over the training set (with fixed sub-sampling rate), until
convergence. In our evaluations, three passes sufficed. Then, we obtain
a distribution of the residual ${\bf S}$ over the entire training set. Next is the
recovery phase, in which we sub-sample a small fraction of the entries of each
successive column $t$, solve for the reconstruction coefficients ${\bf W}(\cdot,t)$
in the basis ${\bf U}$ by least-squares, and then add random residuals using
parameters estimated during training. After that, we proceed exactly as in a
normal permutation testing, to recover the statistics. 

{\em Bias-Variance tradeoff}. By using a very sparse subsampling method, there
is a bias-variance dilemma in estimating ${\bf S}$.
That is, if we use the entire matrix $\PB$ to estimate ${\bf U}$, ${\bf
W}$ and ${\bf S}$, we will obtain reliable estimates of ${\bf S}$. But,
there is an overfitting problem: the least-squares objective used in fitting
${\bf W}(\cdot,t)$ to such a small sample of entries is likely to grossly
underestimate the variance of ${\bf S}$ compared to where we use
the entire matrix; (the sub-sampling problem is not nearly as over-constrained
as for the whole matrix). This sampling artifact reduces the
apparent variance of ${\bf S}$, and induces a bias in the distribution of the sample
maximum, because extreme values are found less frequently. This sampling
artifact has the effect of `shifting' the distribution of the sample maximum
towards 0. We correct for this bias by estimating the amount of the shift during
the training phase, and then shifting the recovered sample max distribution by
this estimated amount.

\section{Analysis}
\label{sec:model}

We now discuss two results which show that as long as the variance of the
residual is below a certain level, we can recover the distribution of
the sample maximum. Recall from \eqref{eq:matrix-completion} that
for low-rank matrix completion methods to be applied we must assume that
the permutation matrix $\PB$ can be decomposed into a low-rank component plus a
high-rank residual matrix ${\bf S}$:
\begin{equation} \label{eq:def}
  \PB = {\bf UW} + {\bf S},
\end{equation}
where ${\bf U}$ is a $v \times r$ orthogonal matrix that spans the $r \ll
\min(v,t)$ -dimensional column subspace of $\PB$, and ${\bf W}$ is the
corresponding coefficient matrix. We can then treat the residual ${\bf S}$ as a
random matrix whose entries are \iid ~zero-mean Gaussian with variance $\sigma^2$. 
We arrive at our first result by analyzing how the low-rank portion of 
$\PB$'s singular spectrum interlaces with the contribution coming from the residual
by treating $\PB$ as a low-rank perturbation of a random matrix. 
If this low-rank perturbation is sufficient to dominate the eigenvalues of
the random matrix, then $\PB$ can be recovered with high fidelity at a low
sampling rate \cite{balzano2010, he2012}. Consequently, we can estimate the
distribution of the maximum as well, as shown by our second result.

The following development relies on the observation that the eigenvalues of
$\PB\PB^T$ are the squared singular values of $\PB$. Thus, rather than analyzing
the singular value spectrum of $\PB$ directly, we can analyze the eigenvalues of
$\PB\PB^T$ using a recent result from \cite{benaych2011}. This is important
because in order to ensure recovery of $\PB$, we require that its singular value
spectrum will approximately retain the shape of ${\bf UW}$'s. More precisely, we
require that for some $0 < \delta < 1$,
\begin{equation} \label{eq:cons}
  |\tilde{\phi_{i}} - \phi_{i}| < \delta\phi_{i} \qquad i=1,\ldots,r ; \qquad
  \tilde{\phi_{i}} < \delta\phi_{r} \qquad i=r+1,\ldots,v
\end{equation}
where $\phi_{i}$ and $\tilde{\phi_{i}}$ are the singular values of ${\bf UW}$
and ${\bf P}$ respectively. (Recall that in this analysis $\PB$ is the
perturbation of ${\bf UW}$.) 
Thm. \ref{thm:eigval}
relates the rate at which eigenvalues are perturbed, $\delta$, 
to the parameterization of ${\bf S}$ in terms of $\sigma^2$.
The theorem's principal assumption also relates $\sigma^2$ inversely with the number
of columns of $\PB,$ which is just the number of trials $t$. Note however that
the process may be split up between several matrices $\PB_i$, and the results
can then be combined. For purposes of applying this result in practice we may
then choose a number of columns $t$ which gives the best bound. Theorem \ref{thm:eigval} also assumes that the number of trials $t$ is greater than the
number of voxels $v$, which is a difficult regime to explore empirically. Thus,
our numerical evaluations cover the case where $t < v$, while Thm
\ref{thm:eigval} covers the case where $t$ is larger.

From the definition of ${\bf P}$ in \eqref{eq:def}, we have,
\begin{equation}\label{eq:exp}
{\bf PP}^T = {\bf UWW}^T{\bf U}^T + {\bf SS}^T + {\bf UWS}^T + {\bf SW}^T{\bf U}^T.
\end{equation}
We first analyze the change in eigenvalue structure of ${\bf SS}^T$ when perturbed by 
${\bf UWW}^T{\bf U}^T$, (which has $r$ non-zero eigenvalues). The influence of
the cross-terms (${\bf UWS}^T$ and ${\bf SW}^T{\bf U}^T$) is addressed later.
Thus, we have the following theorem.

\vspace{3mm}
\begin{theorem}\label{thm:eigval}
Denote that $r$ non-zero eigenvalues of
${\bf Q} = {\bf UWW}^T{\bf U}^T \in \R^{v \times v}$ by 
$\lambda_{1}\geq\lambda_{2}\geq,\ldots,\lambda_{r}>0$; and let 
${\bf S}$ be a $v \times t$ random matrix such that ${\bf S}_{i,j} \sim {\mathcal
N}(0, \sigma^2)$, with unknown $\sigma^2$. As $v,t \to \infty$ such that
$\frac{v}{t} \ll 1$, the eigenvalues $\tilde{\lambda_{i}}$ of the perturbed
matrix ${\bf Q} + {\bf SS}^T$ will satisfy 
\begin{equation} \tag{$\star$} \label{modcons} 
    |\tilde{\lambda_{i}} - \lambda_{i}| < \delta\lambda_{i} \qquad i=1,\ldots,r;
  \qquad \tilde{\lambda_{i}} < \delta\lambda_{r} \qquad i=r+1,\ldots,v
\end{equation} 
for some $0 < \delta < 1$, whenever $\sigma^2 <
\frac{\delta\lambda_{r}}{t}$
\end{theorem}

\begin{proof}
The first half of the proof emulates Theorem $2.1$ from \cite{benaych2011}. 
Consider the matrix ${\bf X} = \sqrt{t}{\bf S}$. 
By the structure of ${\bf S}$, each entry of ${\bf X}$ is \iid Gaussian with zero--mean and variance $\sigma^2t$. 
Let ${\bf Y} = \frac{1}{t}{\bf XX}^T$ and denote its ordered eigenvalues as $\gamma_{i}, i=1,\dots,v$ (large to small). 
Consider the random spectral measure 
\begin{center}
  $\mu_{v}(A) = \frac{1}{v}\#\{\gamma_{i} \in A\}$, \hspace{5mm} $A \subset \mathbb{R}$
\end{center}
The Marchenko--Pastur law \cite{marvcenko} states that as $v,t \to \infty$ such that $\frac{v}{t} \leq 1$, the random measure $\mu_{v} \to \mu$, where $d\mu$ is given by
\begin{center}
  $d\mu(a) = \frac{1}{2\pi\sigma^2t\gamma a} \sqrt{(\gamma_{+}-a)(a-\gamma_{-})} 
    {\bf 1}_{[\gamma_{-},\gamma_{+}]}da$
\end{center}
where $\gamma = \frac{v}{t}$. 
Here ${\bf 1}_{[\gamma_{-},\gamma_{+}]}$ is an indicator function that is non--zero on $[\gamma_{-},\gamma_{+}]$. 
$\gamma_{\pm} = \sigma^2t (1 \pm \sqrt{\gamma})^2$ are the extreme points of the support of $\mu$. 
It is well known that the extreme eigenvalues converge almost surely to $\gamma_{\pm}$ \cite{Edelman}. 
Since $v,t \to \infty$ and $\gamma = \frac{v}{t} \ll 1$, the length of $[\gamma_{-},\gamma_{+}]$ is much smaller than the values in it. 
Hence we have,
\begin{center}
  $\gamma_{\pm} \sim \sigma^2t (1 \pm 2\sqrt{\gamma})$ \hspace{3mm};\hspace{3mm} 
    $\sqrt{(\gamma_{+}-a)(a-\gamma_{-})} \ll a$
\end{center}
and the new $d\mu(a)$ is given by
\begin{align*}
  d\mu(a) &= \frac{\sqrt{(\sigma^2t(1+2\sqrt{\gamma})-a)(a-\sigma^2t(1-2\sqrt{\gamma}))}}{2\pi\gamma\sigma^4t^2} {\bf 1}_{[\sigma^2t(1-2\sqrt{\gamma}),\sigma^2t(1+2\sqrt{\gamma})]} da \\
  &= \frac{1}{2\pi\gamma\sigma^4t^2} \sqrt{4\gamma\sigma^4t^2 - (a - \sigma^2t)^2} {\bf 1}_{[\sigma^2t(1-2\sqrt{\gamma}),\sigma^2t(1+2\sqrt{\gamma})]} da
\end{align*}
The form we have derived for $d\mu(a)$ shares some similarities with
$d\mu_{X}(x)$ in Section $3.1$ of \cite{benaych2011}. 
The analysis in \cite{benaych2011} takes into account the phase transition of extreme eigen values. 
This is done by imitating a time--frequency type analysis on compact support of extreme spectral measure i.e. using Cauchy transform. 
For our case, the Cauchy transform of $\mu(a)$ is
\begin{align*}
  G_{\mu}(z) & = \frac{1}{2\gamma\sigma^4t^2} \left(z - \sigma^2t - sgn(z)\sqrt{(z - \sigma^2t)^2 - 4\gamma\sigma^4t^2} \right) \\
  & \hspace{10mm} \text{for} \hspace{3mm} z \in (\infty,\sigma^2t(1-2\sqrt{\gamma}))\cup(\sigma^2t(1+2\sqrt{\gamma}),\infty)
\end{align*}
Since we are interested in the asymptotic eigen values (and $\gamma \ll 1$), $G_{\mu}(\gamma_{\pm})$ and the functional inverse $G^{-1}_{\mu}(\theta)$ are
\begin{center}
$G_{\mu}(\gamma_{+}) = \frac{1}{\sigma^2t\sqrt(\gamma)} \hspace{2mm};\hspace{2mm} G_{\mu}(\gamma_{-}) = -\frac{1}{\sigma^2t\sqrt(\gamma)} \hspace{2mm};\hspace{2mm} G^{-1}_{\mu}(\theta) = \sigma^2t + \frac{1}{\theta} + \gamma\sigma^4t^2\theta$
\end{center}
Hence, the asymptotic behavior of the eigen values of perturbed matrix ${\bf Q} + {\bf SS}^T$ is (observing that ${\bf SS}^T = {\bf Y}$ and ${\bf Q}$ has $r$ non--zero positive eigen values)
\begin{align*}\tag{$\ast$}\label{purteig}
\tilde{\lambda_{i}} (i=1,\ldots,r) \hspace{2mm} &\approx \hspace{2mm} \begin{cases} \lambda_{i} + \sigma^2t + \frac{\gamma\sigma^4t^2}{\lambda_{i}} & \text{for} \hspace{3mm} \lambda_{i} > \gamma\sigma^2t \\ \gamma\sigma^2t & \text{else} \end{cases} \\
\tilde{\lambda_{i}} (i=r+1,\ldots,v) \hspace{2mm} &\approx \hspace{2mm} \sigma^2t(1-2\sqrt(\gamma))
\end{align*}
With $\tilde{\lambda_{i}}, i=1,\ldots,v $ in hand, we now bound the unknown variance $\sigma^2$ such that $(\star)$ is satisfied. 
We only have two cases to consider,
\begin{center}
(1) \hspace{1mm} $\lambda_{i} > \gamma\sigma^2t$ , $i=1,\ldots,r$ \hspace{6mm}
(2) \hspace{1mm} $\lambda_{i} \leq \gamma\sigma^2t$ , $i=k,\ldots,r$ (for some $k \geq 1$)
\end{center}
We constrain the unknown $\sigma^2$ such that case (2) does not arise. 
Substituting for $\tilde{\lambda_{i}}$`s from $(\ast)$ in $(\star)$, we get,
\begin{center}
$\sigma^2t + \frac{\gamma\sigma^4t^2}{\lambda_{i}} < \delta\lambda_{i}$ \hspace{3mm};\hspace{3mm} $\lambda_{i} > \gamma\sigma^2t$ \hspace{3mm};\hspace{3mm} $\sigma^2t(1-2\sqrt(\gamma)) < \delta\lambda_{r}$
\end{center}
These inequalities will hold when $\sigma^2 < \frac{\delta\lambda_{r}}{t}$ (since $\gamma = \ll 1$, $\delta < 1$ and $\lambda_{1}\geq\lambda_{2}\geq,\ldots,\lambda_{r}$). 
\end{proof}

Note that the missing cross-terms would not change the result of Theorem
$\ref{thm:eigval}$ drastically, because ${\bf UW}$ has $r$ non-zero singular
values and hence ${\bf UW S}^T$ is a low-rank projection of a low-variance 
random matrix, and this will clearly be dominated by either of the other terms.
Having justified the model in \eqref{eq:def}, the following thorem 
shows that the empirical distribution of the maximum Null statistic approximates
the true distribution.

\vspace{10pt}
\begin{theorem}
Let $m_t = \max_i \PB_{i,t}$ be the maximum observed test statistic at
permutation trial $t$, and similarly let $\hat{m}_t = \max_i \hat{\PB}_{i,t}$ be
the maximum reconstructed test statistic.
Further, let the maximum reconstruction error be $\epsilon$, such that
$|\PB_{i,t} - \hat{\PB}_{i,t}| \le \epsilon$. Then, for any real number $k>0$, we have,
\[
    \mbox{Pr}\left[ m_t - \hat{m}_t - (b - \hat{b}) > k\epsilon \right] <
\frac{1}{k^2}
\]
where $b$ is the bias term described in Section 2, and $\hat{b}$ is its estimate
from the training phase.
\end{theorem}

\begin{proof}
Recall that there is a bias term in estimating the distribution of the maximum
which must be corrected for
this is because $\mbox{var}(\hat{S})$ underestimates $\mbox{var}(S)$ due to the bias/variance tradeoff. 
Let $b$
be this difference:
\[
  b = {\mathbb E}_t \left[ \max_i \PB_{i,t} \right] -
  {\mathbb E}_t \left[ \max_i \hat{\PB}_{i,t} \right].
\]
Further, recall that we estimate $b$ by taking the difference of mean sample
maxima between observed and reconstructed test statistics
over the training set, giving $\hat{b}$, which is an unbiased estimator of
$b$ ---  it is unbiased because a difference in sample means is an unbiased
estimator of the difference of two expectations.

Let $\delta_t = m_t - \hat{m}_t$. To show the result we must derive a
concentration bound on $\delta_t$,
which we will do by applying Chebyshev's inequality. In order to do so, we require an
expression for the mean and variance of $\delta_t$. First, we derive an
expression for the mean. Taking the expectation over $t$ of $m_t - \hat{m}_t$ we
have,
\begin{align*}
    {\mathbb E}_t \left[ m_t - \hat{m}_t \right] & =
    {\mathbb E}_t \left[ \max_i \PB_{i,t} - \max_i\hat{\PB}_{i,t} - \hat{b} \right] \\
        & = {\mathbb E}_t \left[ \max_i \PB_{i,t} \right] - 
        {\mathbb E}_t \left[ \max_i \hat{\PB}_{i,t} \right] - \hat{b} \\
        & = b - \hat{b}
\end{align*}
where the second equality follows from the linearity of expectation. 

Next, we require an expression for the variance of $\delta_t$.
Let $i$ be the index at which the maximum observed test statistic occurs
for permutation trial $t$, and likewise let $j$ be the index at which the
maximum reconstructed test statistic occurs. Thus we have,
\begin{align*}
    \PB_{i,t} \le& ~ \hat{\PB}_{i,t} + \epsilon ~\le \hat{\PB}_{j,t} + \epsilon\\
    \PB_{i,t} \ge& ~ \PB_{j,t} \qquad \ge \hat{\PB}_{j,t} - \epsilon,\\
\end{align*} 
and so we have that
\[
    |m_t - \hat{m}_t| < 2\epsilon
\]
and so
\[
    \mbox{var}(m_t - \hat{m_t}) \le \epsilon^2.
\]

Applying Chebyshev's bound,
\[
    \mbox{Pr}\left[ m_t - \hat{m}_t - (b - \hat{b}) > k\epsilon \right] < \frac{1}{k^2}
\]
which completes the proof.
\end{proof}

\section{Experimental evaluations}
\label{sect-experiments}

Our experimental evaluations include four separate neuroimaging datasets of
Alzheimer's Disease (AD) patients, cognitively healthy age-matched controls
(CN), and in some cases Mild Cognitive Impairment (MCI) patients. The first of
these is the Alzheimer's Disease Neuroimaging Initiative (ADNI) dataset, a
nation-wide multi-site study. ADNI is a landmark study sponsored by the NIH,
major pharmaceuticals and others to determine the extent to which multimodal
brain imaging can help predict onset, and monitor progression of AD. The others
were collected as part of other studies of AD and MCI.
We refer to these datasets as Dataset A---D. 
Their demographic characteristics are as follows:
\begin{inparaenum}[{Dataset} A: ]
\item $40$ subjects, AD vs. CN, median age : $76$;
\item $50$ subjects, AD vs. CN, median age : $68$;
\item $55$ subjects, CN vs. MCI, median age : $65.16$;
\item $70$ subjects, CN vs. MCI, median age : $66.24$.
\end{inparaenum}

Our evaluations focus on three main questions: 
\begin{inparaenum}[\bfseries (i)]
\item Can we recover an acceptable approximation of the maximum statistic Null distribution from an approximation of the permutation test matrix?
\item What degree of computational speedup can we expect at various subsampling rates, and how does this affect the trade-off 
  with approximation error?
\item How sensitive is the estimated $\alpha$-level threshold with respect to the recovered Null distribution?
\end{inparaenum}
In all our experiments, the rank estimate for subspace tracking (to construct the low--rank basis ${\bf U}$) was taken as the number of subjects.

\subsection{Can we recover the Maximum Null?}
\label{sec:null}
Our experiments suggest that our model can recover the maximum Null. 
We use Kullback--Leibler (KL) divergence and Bhattacharya Distance (BD) to compare the estimated maximum Null from our model to the true one. 
We also construct a ``Naive--Null'', where the subsampled statistics are pooled and the Null distribution is constructed with no further processing (i.e., completion).
Using this as a baseline, Fig. \ref{fig:null} shows the KL and BD values obtained from three datasets, at $20$ different 
sub-sampling rates (ranging from $0.1\%$ to $10\%$).
Note that our model involves a training module where the approximate `bias' of residuals is estimated. 
This estimation is prone to 
noise (for example, number of training frames).
Hence Fig. \ref{fig:null} also shows the error bars pertaining to $5$ realizations on the $20$ sampling rates.
The first observation from Fig. \ref{fig:null} is that both KL and BD measures of the recovered Null to the true 
distribution are $< e^{-5}$ for sampling rates more than $0.4\%$.
\begin{figure}[!hb]\centering
\begin{tabular}{ccc}
\subfloat[Dataset A]{\includegraphics[width=0.3\linewidth]{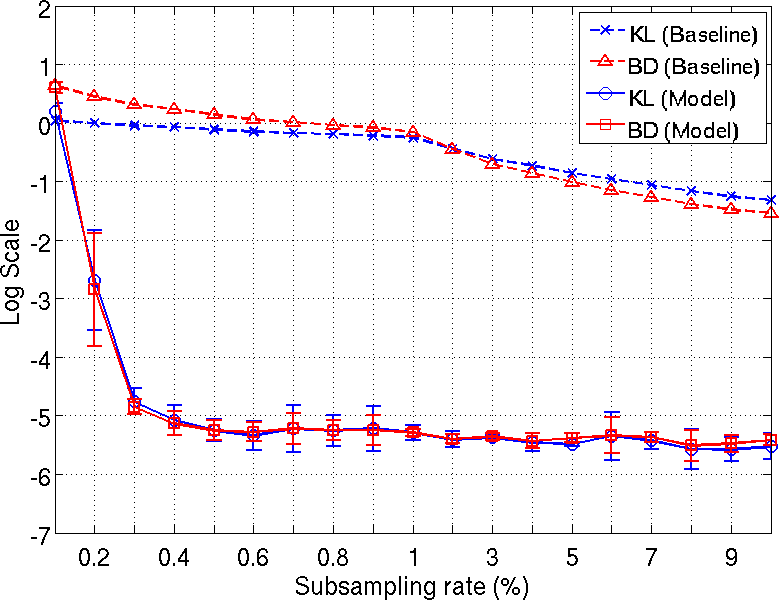}} &
\subfloat[Dataset B]{\includegraphics[width=0.3\linewidth]{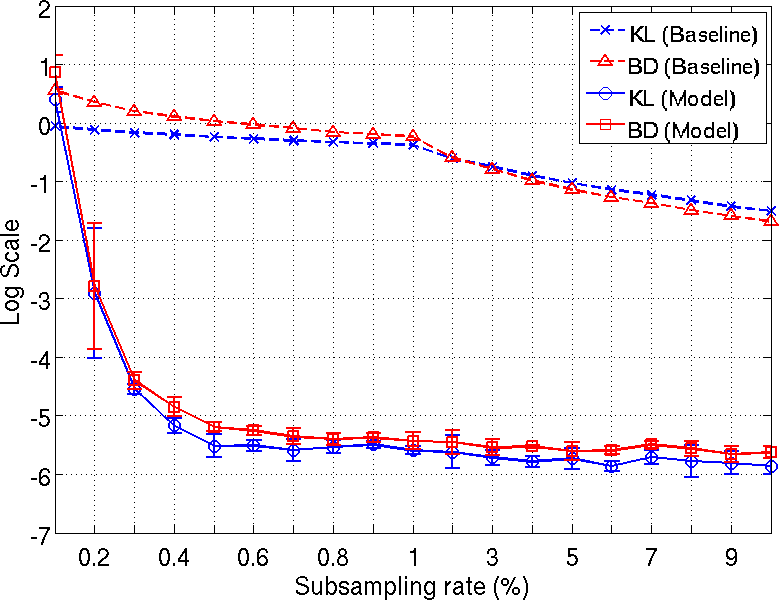}} &
\subfloat[Dataset C]{\includegraphics[width=0.3\linewidth]{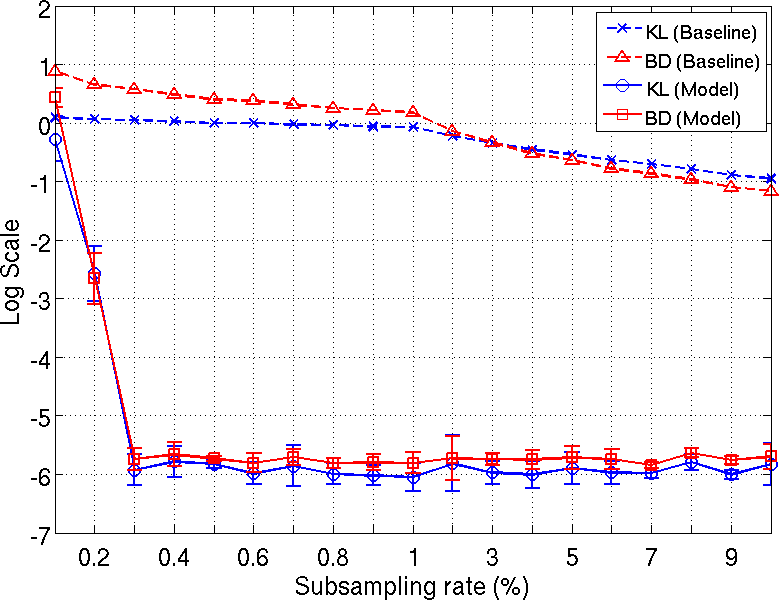}} 
\end{tabular}
\caption{\footnotesize \label{fig:null} KL (blue) and BD (red) measures between the true max Null distribution (given by the full matrix ${\bf P}$) and that recovered by our method (thick lines), along with the baseline naive subsampling method (dotted lines). Results for Datasets A, B, C are shown here. Plot for Dataset D is in the extended version of the paper. }
\end{figure}
This suggests that our model recovers both the shape (low BD) and position (low KL) of the null to high accuracy 
at extremely low sub-sampling.
We also see that above a certain minimum subsampling rate ($\sim 0.3\%$), the KL and BD do not change drastically as the rate is increased.
This is expected from the theory on matrix completion where after observing a minimum number of data samples, adding in new 
samples does not substantially increase information content.
Further, the error bars (although very small in magnitude) of both KL and BD show that the recovery is noisy. 
We believe this is due to the approximate estimate of bias from training module.
\begin{figure}[!h]\centering
\includegraphics[width=0.6\linewidth]{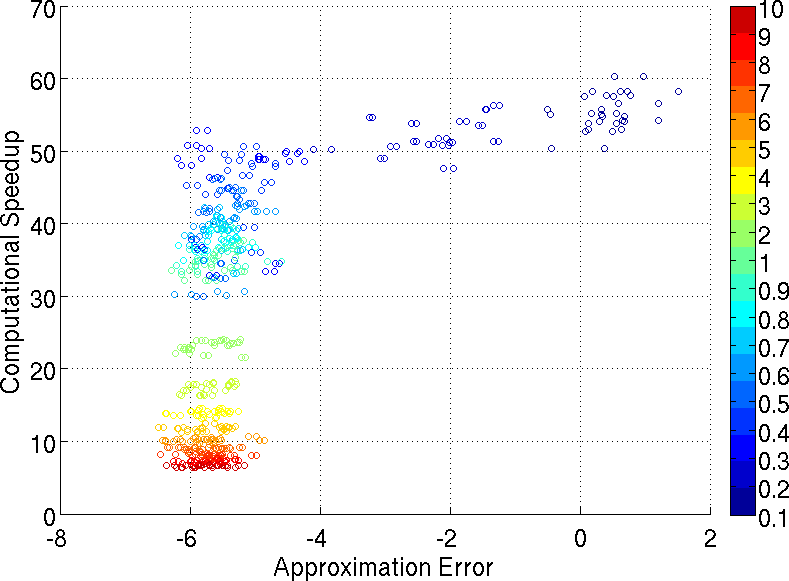}
\caption{\footnotesize \label{fig:scatter} Scatter plot of computational speedup vs. KL. The plot corresponds to the $20$ different samplings on all $4$ datasets (for $5$ repeated set of experiments) and the colormap is from $0.1\%$ to $10\%$ sampling rate. The x--axis is in log scale.}
\end{figure}

\vspace{-5pt}
\subsection{What is the computational speedup?}
\label{sec:comp}
Our experiments suggest that the speedup is substantial.
Figs. \ref{fig:comp} and \ref{fig:scatter} compare the time taken to perform the complete permutation testing to that of our model. 
The three plots in Fig. \ref{fig:comp} correspond to the datasets used in Fig. \ref{fig:null}, in that order.
Each plot contains $4$ curves and represent the time taken by our model, the corresponding sampling and GRASTA \cite{he2012} 
recovery (plus training) times and the total time to construct the entire matrix ${\bf P}$ (horizontal line).
And Fig. \ref{fig:scatter} shows the scatter plot of computational speedup vs. KL divergence (over $3$ repeated set of experiments on all the datasets and sampling rates).
Our model achieved at least $30$ times decrease in computation time in the low sampling regime ($< 1\%$). 
Around $0.5\%-0.6\%$ sub-sampling (where the KL and BD are already $< e^{-5}$),
the computation speed-up factor averaged over all datasets was $45 \times$.
This shows that our model achieved good accuracy (low KL and BD) together with high computational speed up in tandem, especially, for $0.4\%-0.7\%$ sampling rates. 
However note from Fig. \ref{fig:scatter} that there is a trade--off between the speedup factor and approximation error (KL or BD).
Overall the highest computational speedup factor achieved at a recovery level of $e^{-5}$ on KL and BD is around $50$x (and this occured around $0.4\%-0.5\%$ sampling rate, refer to Fig. \ref{fig:scatter}).
It was observed that a speedup factor of upto $55\times$ was obtained for
Datasets C and D at $0.3\%$ subsampling, where the KL and BD were as low as
$e^{-5.5}$ (refer to Fig. \ref{fig:null} and the extended version of the paper).
\begin{figure}[!tb]\centering
{\footnotesize
\begin{tabular}{ccc}
\subfloat[Speedup (at $0.4\%$) is $45.1$]{\includegraphics[width=0.3\linewidth]{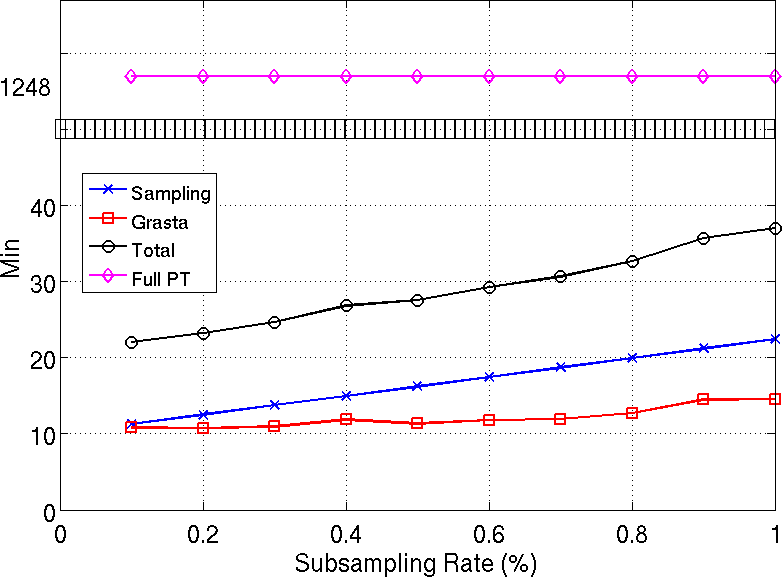}} &
\subfloat[Speedup (at $0.4\%$) is $45.6$]{\includegraphics[width=0.3\linewidth]{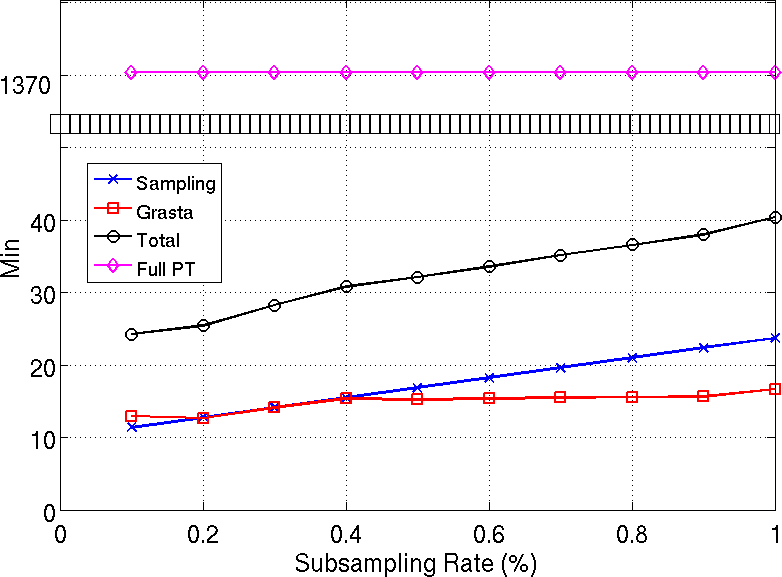}} &
\subfloat[Speedup (at $0.4\%$) is $48.5$]{\includegraphics[width=0.3\linewidth]{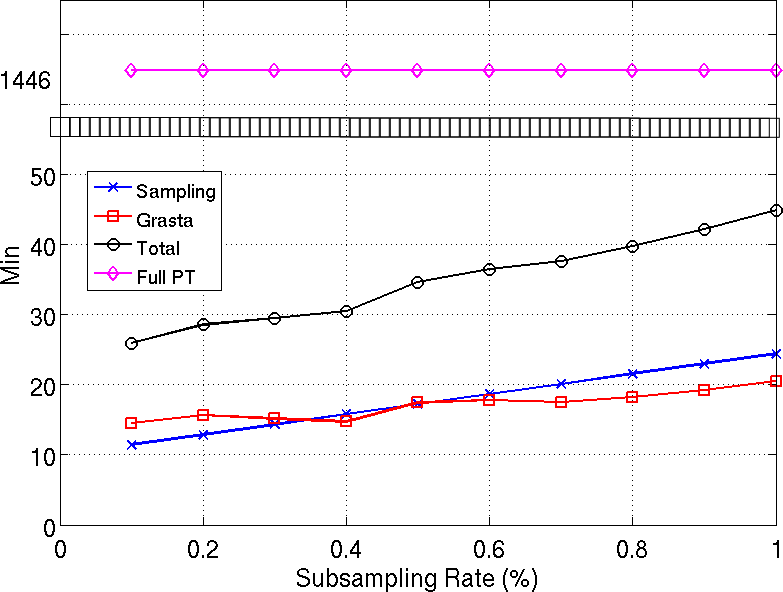}} 
\end{tabular}
}
\vspace{-3pt}
\caption{\footnotesize \label{fig:comp} Computation time (in minutes) of our model compared to that of computing the entire matrix ${\bf P}$. Results are for the same three datasets as in Fig. \ref{fig:null}. Please find the plot for Dataset D in the extended version of the paper. The horizontal line (magenta) shows the time taken for computing the full matrix ${\bf P}$. The other three curves include : subsampling (blue), GRASTA recovery (red) and total time taken by our model (black). Plots correspond to the low sampling regime ($<1\%$) and note the jump in y--axis (black boxes). For reference, the speedup factor at $0.4\%$ sampling rate is reported at the bottom of each plot.}
\vspace{-16pt}
\end{figure}

\subsection{How stable is the estimated $\alpha$-threshold (clinical significance)?}
\label{sec:stable}
Our experiments suggest that the threshold is stable. 
Fig. \ref{fig:alphas} and Table \ref{tab:alphas} summarize the clinical significance of our model.
Fig. \ref{fig:alphas} show the error in estimating the true max threshold, at $1-\alpha=0.95$ level of confidence.
The x--axis corresponds to the $20$ different sampling rates used and y--axis shows the absolute difference of thresholds in log scale.
\begin{figure}[!tb]\centering
\begin{tabular}{ccc}
\subfloat[Datasets A, B]{\includegraphics[width=0.325\linewidth]{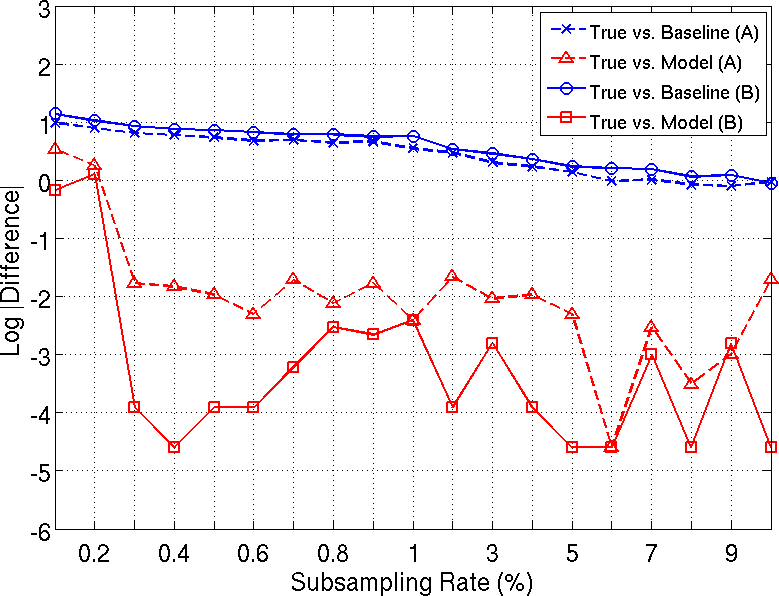}} & \quad \quad
\subfloat[Datasets C, D]{\includegraphics[width=0.325\linewidth]{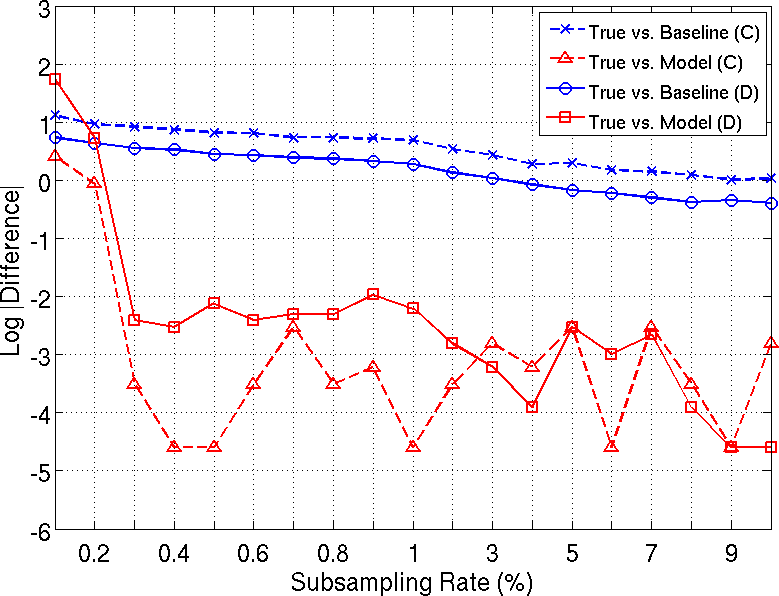}}
\end{tabular}
\caption{\footnotesize \label{fig:alphas} Error of estimated $t$ statistic thresholds (red) for the $20$ different subsampling rates on the four Datasets. The confidence level is $1-\alpha=0.95$. The y-axis is in log--scale. For reference, the thresholds given by baseline model (blue) are included. Note that each plot corresponds to two datasets.}
\end{figure}
\begin{table}
  \centering
  {\footnotesize
  \begin{tabular}[c]{|c|c|c|c|c|c|}
   \hline \cellcolor[gray]{0.85}{Data} & \cellcolor[gray]{0.85}{Sampling} & \multicolumn{4}{>{\columncolor[gray]{0.85}}c|}{$1-\alpha$ level} \\
  \cellcolor[gray]{0.85}{name} & \cellcolor[gray]{0.85}{rate} & \cellcolor[gray]{0.85}{$0.95$} & \cellcolor[gray]{0.85}{$0.99$} & \cellcolor[gray]{0.85}{$0.995$} & \cellcolor[gray]{0.85}{$0.999$} \\ \hline
   \multirow{2}{*}{$A$} & $0.3\%$ & $0.16$ & $0.11$ & $0.14$ & $0.07$ \\ 
   & $0.5\%$ & $0.13$ & $0.08$ & $0.10$ & $0.03$ \\ \hline
   \multirow{2}{*}{$B$} & $0.3\%$ & $0.02$ & $0.05$ & $0.03$ & $0.13$ \\ 
   & $0.5\%$ & $0.02$ & $0.07$ & $0.08$ & $0.04$ \\ \hline
   \multirow{2}{*}{$C$} & $0.3\%$ & $0.04$ & $0.13$ & $0.21$ & $0.20$ \\ 
   & $0.5\%$ & $0.01$ & $0.07$ & $0.07$ & $0.05$ \\ \hline
   \multirow{2}{*}{$D$} & $0.3\%$ & $0.08$ & $0.10$ & $0.27$ & $0.31$ \\
   & $0.5\%$ & $0.12$ & $0.13$ & $0.25$ & $0.22$ \\ \hline
  \end{tabular}}
  \caption{\footnotesize \label{tab:alphas} Errors of estimated $t$
statistic thresholds on all datasets at two different subsampling rates.}
\end{table}
Observe that for sampling rates higher than $3\%$, the mean and maximum differences was $0.04$ and $0.18$.
Note that the binning resolution of max.statistic used for constructing the Null was $0.01$.
These results show that not only the global shape of the maximum Null distribution is estimated to high accuracy (see Section \ref{sec:null}) but also the shape and area in the tail.
To support this observation, we show the absolute differences of the estimated thresholds on all the datasets at $4$ different $\alpha$ levels in Table \ref{tab:alphas}.
The errors for $1-\alpha=0.95,0.99$ are at most $0.16$. The increase in error for $1-\alpha>0.995$ is a sampling artifact and is expected.
Note that in a few cases, the error at $0.5\%$ is slightly higher than that at $0.3\%$ suggesting that the recovery is noisy (see Sec. \ref{sec:null} and the errorbars of Fig. \ref{fig:null}). 
Overall the estimated $\alpha$-thresholds are both faithful and stable.

\section{Conclusions and future directions}
\label{sect-conclusion}

In this paper, we have proposed a novel method of efficiently approximating the
permutation testing matrix by first estimating the major singular vectors, then filling
in the missing values via matrix completion, and finally estimating the
distribution of residual values. Experiments on four different neuroimaging
datasets show that we can recover the distribution of the maximum Null statistic
to a high degree of accuracy, while maintaining a computational speedup factor of
roughly $50\times$.
While our focus has been on neuroimaging problems, we note that multiple testing
and False Discovery Rate (FDR) correction are important issues in genomic and
RNA analyses, and our contribution may offer enhanced leverage to existing
methodologies which use permutation testing in these settings\cite{storey}.

\vspace{3mm}
{\small {\bf Acknowledgments:} We thank Robert Nowak, Grace Wahba, Moo K. Chung and the anonymous reviewers for their helpful comments, and Jia Xu for helping with a preliminary implementation of the model.
This work was supported in part by NIH R01 AG040396; NSF CAREER grant 1252725; NSF RI 1116584; 
Wisconsin Partnership Fund; UW ADRC P50 AG033514; UW ICTR 1UL1RR025011 and a Veterans Administration Merit Review Grant I01CX000165. 
Hinrichs is supported by a CIBM post-doctoral fellowship via NLM grant 2T15LM007359. 
The contents do not represent views of the Dept. of Veterans Affairs or the United States Government. }

\newpage

\section*{Fig 1. : All $4$ datasets}

\begin{figure}[!htb]\centering
\begin{tabular}{ccc}
\subfloat[Dataset A]{\includegraphics[width=0.45\linewidth]{ADNI_KLBD}} &
\subfloat[Dataset B]{\includegraphics[width=0.45\linewidth]{ADRC_KLBD}} \\
\subfloat[Dataset C]{\includegraphics[width=0.45\linewidth]{PREDICT_KLBD}} &
\subfloat[Dataset D]{\includegraphics[width=0.45\linewidth]{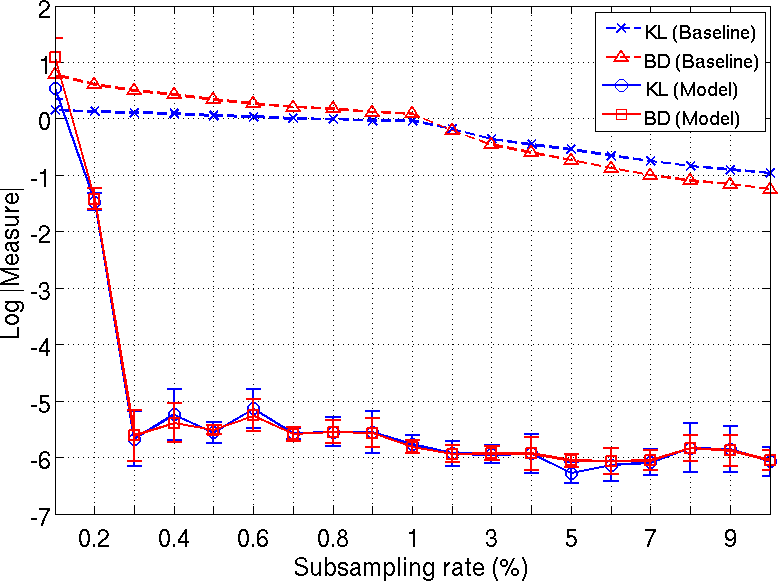}} 
\end{tabular}
\caption{\footnotesize \label{fig:null} KL (blue) and BD (red) measures between the true max. null distribution (given by the full matrix ${\bf P}$) and that recovered by our method (thick lines), along with the baseline naive subsampling method (dotted lines). Each plot corresponds to one of the four datasets used in our evaluations. Note that the y--axis is in log scale.}
\end{figure}

\section*{Fig 2. : All $4$ datasets}

\begin{figure}[!htb]\centering
\begin{tabular}{ccc}
\subfloat[Speedup (at $0.4\%$) - $45.1$]{\includegraphics[width=0.45\linewidth]{ADNI_time}} &
\subfloat[Speedup (at $0.4\%$) - $45.6$]{\includegraphics[width=0.45\linewidth]{ADRC_time}} \\
\subfloat[Speedup (at $0.4\%$) - $48.5$]{\includegraphics[width=0.45\linewidth]{PREDICT_time}} & 
\subfloat[Speedup (at $0.4\%$) - $56.4$]{\includegraphics[width=0.45\linewidth]{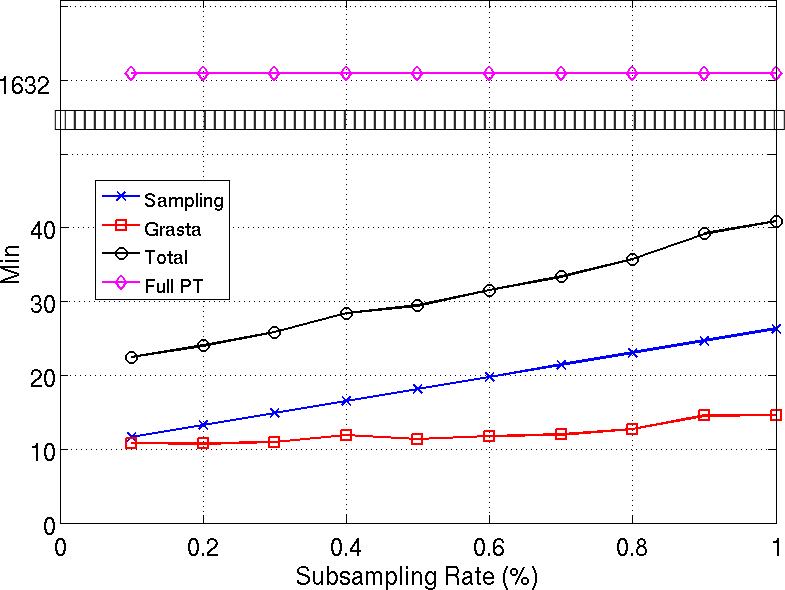}} 
\end{tabular}
\caption{\footnotesize \label{fig:comp} Computation time of our model compared
to that of computing the entire matrix ${\bf P}$. Each plot corresponds to one
of the four datasets A, B, C and D (in that order). The horizontal line
(magenta) shows the time taken for computing the full matrix ${\bf P}$. The
other three curves include : subsampling (blue), GRASTA recovery (red) and total
time taken by our model (black). Plots correspond to the low sampling regime
($<1\%$) and note the jump in y--axis (black boxes).
For reference the speedup at $0.4\%$ sampling rate is reported at the
bottom of each plot.}
\end{figure}

{
    \footnotesize
    \bibliographystyle{unsrt}
    \bibliography{LRMC_perm_testing}
}

\end{document}